\documentclass[12 pt]{amsart}

\usepackage{xpatch} 
\usepackage{ifthen}

\usepackage[]{amsmath} 
\usepackage{amssymb}
\usepackage{mathtools} 
\usepackage{amsfonts}
\usepackage{mathrsfs}
\usepackage[cal=cm, scr=dutchcal, bb=libus]{mathalpha}

\usepackage{amsthm}
   \theoremstyle{plain} 
      
      \newtheorem{theorem}{Theorem}[section]
      \newtheorem{proposition}[theorem]{Proposition}
   \theoremstyle{definition}
      \newtheorem{lemma}[theorem]{Lemma}
      
      \newtheorem{definition}[theorem]{Definition}
      \newtheorem{construction}[theorem]{Construction}
      
   \theoremstyle{remark}
      \newtheorem{example}[theorem]{Example}
      \newtheorem{remark}[theorem]{Remark}
      
\usepackage{imakeidx}
\usepackage[style=alphabetic, sorting=nyt]{biblatex}
   \bibliography{main}
\usepackage{hyperref}
   \hypersetup{colorlinks=true, linkcolor=blue, citecolor=teal, urlcolor=teal,
   linktocpage, breaklinks=true}
\usepackage{xurl}

\usepackage{tikz-cd} 

\usepackage[euler]{textgreek} 
\usepackage{xcolor} 
   \definecolor{myred}{rgb}{0.5,0,0} 
   \definecolor{mygray}{rgb}{0.333,0.333,0.333} 
\usepackage{soul} 
\usepackage{lipsum} 

\makeatletter
   \renewcommand{\eqref}[1]{\textup{\eqreftagform@{\ref{#1}}}}
   \let\eqreftagform@\tagform@
   \def\tagform@#1{%
   \maketag@@@{\makebox[1sp][r]{\hspace{0em}
      (\ignorespaces#1\unskip\@@italiccorr)}}%
   }
\makeatother
\usepackage[bold=0.3]{xfakebold} 

\newcommand{\nsd}[1]{_{\bar{#1}}} 
\newcommand{\nsu}[1]{^{\bar{#1}}} 
\newcommand\eps{\varepsilon}
\newcommand\kap{(\varepsilon - \ixi \psi)}
\makeatletter
\newcommand{\pushright}[1]{%
   \ifmeasuring@#1\else\omit\hfill$\displaystyle#1$\fi\ignorespaces}
\newcommand{\pushleft}[1]{%
   \ifmeasuring@#1\else\omit$\displaystyle#1$\hfill\fi\ignorespaces}
\makeatother
\newcommand{\term}[1]{{\bolder{#1}}}
\newcommand\esig{\underline \rho}

\newcommand{\cl}[1]{#1_{\textup{cl}}} \newcommand{\sheaf}[1][]%
   {\ifthenelse{\equal{#1}{}}{\mathcal O}{\mathcal O_{#1}}}
\newcommand{\ld}[1][]%
   {\ifthenelse{\equal{#1}{}}{\mathcal L}{\mathcal L_{#1}}}
\newcommand{\partdev}[1]{\frac{\partial}{\partial #1}}
\newcommand{\lvec}[1]%
   {\vphantom{#1}\smash{%
      \overset{\raisebox{-1 ex}{$\scriptscriptstyle\leftarrow$}}{#1}}}
\newcommand{\rvec}[1]%
   {\vphantom{#1}\smash{%
      \overset{\raisebox{-1 ex}{$\scriptscriptstyle\rightarrow$}}{#1}}}
\newcommand{\smooth}[1]{C^\infty\left(#1\right)}

\newcommand\M{\mathcal M}
\newcommand\MN{\mathcal N}

\newcommand{\intrange}[1]{[\![#1]\!]}
\newcommand\N{\mathbb{N}}
\newcommand\Z[1][]%
   {\ifthenelse{\equal{#1}{}}{\mathbb Z}{\mathbb Z_{#1}}}
\newcommand\R{\mathbb R}

\newcommand{\op}[1]{\textup{#1}}
\newcommand{\bolder}[1]{\setBold #1\unsetBold}
\newcommand{\bm}[1]{\boldsymbol{#1}}
\newcommand\sign[1]{(-1)^{#1}}

\newcommand{\Psia}[1][]{\ifthenelse{\equal{#1}{}}{\Psi}{\Psi^{(#1)}}}
\newcommand{\Psib}[1][]{\ifthenelse{\equal{#1}{}}{\Psi'}{\Psi'^{(#1)}}}

\newcommand{\quot}[2]%
   {\left.{\raisebox{.5 ex} {$#1$}}\middle/%
   {\raisebox{-.5 ex}{$#2$}}\right.}
\newcommand{\abs}[2][]{\ifthenelse{\equal{#1}{}}{\lvert #2\rvert}{\left\lvert #2\right\rvert}}
\newcommand{\ext}[2][]{#2^{\textup{ext}}_{\textup{\tiny{#1}}}}
\newcommand{\antif}{+}
\newcommand{\anti}[1]{ \vphantom{#1} {#1^{\antif}}}
\newcommand{\inv}[1]{\vphantom{#1} {#1^{-1}}}
\newcommand{\dc}[1]{\vphantom{#1} {\overline{#1}}}
\newcommand{\hc}[1]{\vphantom{#1} {#1^{\dagger}}}
\newcommand{\cc}[1]{\vphantom{#1} {#1^{\ast}}}

\newcommand\fgr{\mathcal{F}_{\textup{\tiny GR}}}
\newcommand\fbf{\mathcal{F}_{\textup{\tiny $BF$}}}
\newcommand\sgr{S_{\textup{\tiny GR}}}
\newcommand\sbf{S_{\textup{\tiny $BF$}}}

\newcommand\gbf{G_{\textup{\tiny $BF$}}}
\newcommand\egbf{G_{\textup{\tiny $BF$}}^{\textup{ext}}}

\newcommand\phibf{\Phi_{\textup{\tiny $BF$}}}
\newcommand\omgr{\omega_{\textup{\tiny GR}}}
\newcommand\ombf{\omega_{\textup{\tiny $BF$}}}
\newcommand\omgrf{\omega_{\textup{\tiny GR\textPhi}}}
\newcommand\ombff{\omega_{\textup{\tiny $BF$\textPhi}}}
\newcommand\tgr{\mathscr T_{\textup{\tiny GR}}}
\newcommand\tbf{\mathscr T_{\textup{\tiny $BF$}}}
\newcommand\tgrf{\mathscr T_{\textup{\tiny GR\textPhi}}}
\newcommand\tbff{\mathscr T_{\textup{\tiny $BF$\textPhi}}}

\newcommand\qgr{Q_{\textup{\tiny GR}}}
\newcommand\qbf{Q_{\textup{\tiny $BF$}}}
\newcommand\fgrf{\mathcal{F}_{\textup{\tiny GR\textPhi}}}
\newcommand\fbff{\mathcal{F}_{\textup{\tiny $BF$\textPhi}}}
\newcommand\sgrf{S_{\textup{\tiny GR\textPhi}}}
\newcommand\sbff{S_{\textup{\tiny $BF$\textPhi}}}
\newcommand\lgrf{\mathscr{L}_{\textup{\tiny GR\textPhi}}}
\newcommand\lbff{\mathscr{L}_{\textup{\tiny $BF$\textPhi}}}
\newcommand\gbff{G_{\textup{\tiny $BF$\textPhi}}}

\newcommand\phibff{\Phi_{\textup{\tiny $BF$\textPhi}}}

\newcommand\qgrf{Q_{\textup{\tiny GR\textPhi}}}
\newcommand\qbff{Q_{\textup{\tiny $BF$\textPhi}}}
\newcommand\ixi{\iota_{\xi}}
\newcommand\lxi{\mathcal{L}_{\xi}}

\newcommand\lxig{\mathcal{L}_{\xi}^{\Gamma}}
\def\Dg{D_\Gamma}

\newcommand\dego{\deg_{\Omega}}

\newcommand\gh{\textup{gh}\:}

\numberwithin{equation}{section}

\providecommand{\keywords}[1]
{
  \small	
  \textbf{\textit{Keywords---}} #1
}

\begin{document}

\begin{abstract}
   Three-dimensional supergravity in the Batalin--Vilkovisky formalism is constructed
    by showing that the theory including the Rarita--Schwinger term
    is equivalent to an AKSZ theory.
\end{abstract}

\title[3D Supergravity in the Batalin--Vilkovisky Formalism]{3D Supergravity in the Batalin--Vilkovisky Formalism}
\author[A. S. Cattaneo]{A. S. Cattaneo}
\author[N. Moshayedi]{N. Moshayedi}
\author[A. Smailovic Funcasta]{A. Smailovic Funcasta}
\address{Institut f\"ur Mathematik\\ Universit\"at Z\"urich\\ 
Winterthurerstrasse 190
CH-8057 Z\"urich}
\email[A.~S.~Cattaneo]{cattaneo@math.uzh.ch}
\address{Institut f\"ur Mathematik\\ Universit\"at Z\"urich\\ 
Winterthurerstrasse 190
CH-8057 Z\"urich}
\email[N.~Moshayedi]{nima.moshayedi@math.uzh.ch}
\address{Department of Physics\\ ETH Z\"urich\\ 
Otto-Stern-Weg 1
CH-8093 Z\"urich}
\email[A.~Smailovic Funcasta]{asmailovic@ethz.ch}
\date{December 18, 2024}
\keywords{Supergravity, BV formalism, AKSZ formalism, Chern–-Simons and $BF$ theory, Palatini--Cartan formalism}
\thanks{ASC acknowledges partial support of the SNF Grant No.~200021\_227719 and
of the Simons Collaboration on Global Categorical Symmetries. This research was
(partly) supported by the NCCR SwissMAP, funded by the Swiss National Science
Foundation. This article is based upon work from COST Action 21109 CaLISTA,
supported by COST (European Cooperation in Science and Technology)
(www.cost.eu), MSCA-2021-SE-01-101086123 CaLIGOLA, and MSCA-DN CaLiForNIA - 101119552.}
\maketitle

\tableofcontents

\section{Introduction}
\noindent
In this paper we construct the BV action \cite{Batalin1981} for supergravity in
three dimensions (see, e.g.\ \cite{ACHUCARRO198689,RuizRuiz:1996mm,Castellani:2016ibp,Andrianopoliquantum} and
references therein). 

Our strategy is to start considering the minimal coupling
of spin $\frac 32$ fermions to $BF$ theory (with Lie algebra $\mathfrak{so}(2,1)$), where
the BV action is easy to construct, since we realise
this as an AKSZ theory \cite{A06}. 

Three-dimensional $BF$ theory 
is on-shell
equivalent to gravity in the Palatini--Cartan formalism \cite{Witten:1988hc}, and
the off-shell equivalence in the BV formalism has been constructed in
\cite{A01}. Similarly, a super version of three-dimensional $BF$ theory 
is on-shell
equivalent to supergravity in the Palatini--Cartan formalism \cite{ACHUCARRO198689}.

The novel contribution of this paper is 
the extension of the 
BV transformation of \cite{A01} in the presence of spin $\frac 32$ fermions, see Theorem~\ref{th:BFFtoGRF}.
The advantage is that the rather involved BV structure of three-dimensional supergravity (for which we found no explicit expression in the literature) is obtained from the straightforward BV structure of super $BF$ theory via the AKSZ formulation. The main point of the transformation is that the AKSZ symmetries are expressed as covariant derivatives of ghosts which are Lie algebra valued $0$-forms, whereas in (super)gravity we want to see ghosts for (super)diffeomorphisms explicitly.

One can also observe that in our approach
supergravity is ``discovered,'' in the sense that the resulting BV operator
turns out to contain the local supersymmetry transformations. More precisely, switching to zero all the ghosts but for the fermionic ghost $\eps$, we get
\begin{align*}
    Qe &= 
   \dc \psi \, \rho \eps, \\
   Q\psi &= D\eps,
\end{align*}
where $e$ denotes the dreibein, $\psi$ the spin $\frac 32$ field, $\rho$ the spin representation, and $D$ the covariant derivative (note that the spin connection is not affected by this transformation). Moreover, the ghost for diffeomorphisms accordingly transforms with a term $ \frac 12 \dc \varepsilon \, \inv e (\rho)\varepsilon$, which encodes the fact that two supersymmetry transformations yield a translation.

The BV structure we obtain for supergravity in
three dimensions is considerably simpler than in four, where the BV action is known to require non-linear terms in the antifields (see \cite{BAULIEU1990387} for the half-shell formulation and \cite{cattaneofilarobattino}for the Palatini--Cartan formulation). One reason for this may be that three-dimensional gravity is topological and remains so when adding spin $\frac 32$ fermions.

In this note we focus on the Minkowski signature and on the 
spin $\frac 32$ Majorana representation. However, these results are readily generalised
to the Euclidean case (or to alternative signatures) and to other real 
representations
of $\mathfrak{spin}(2,1)$ or $\mathfrak{spin}(3)$, respectively.

\noindent
\subsection{Motivation} 
The Palatini--Cartan--Holst (PCH) formalism 
was conceived in order to bring
Einstein's general relativity closer to the language of gauge theories,%
\footnote{\ A general review of the Palatini formalism is found in \cite{A04},
one of the PCH formalism in \cite{Cattaneo2018}, and Holst's original work is
\cite{Holst1996}.} and based on the  work in \cite{A12} it was proven
in \cite{A01} that in three dimensions Palatini--Cartan (PC) gravity is strongly
equivalent to a $BF$ theory, result that serves as a basis for the current paper.
Here we will show that this strong equivalence persists even when we incorporate
the Rarita--Schwinger term to the action,\footnote{\ For an introduction to
supersymmetry, see \cite{A13}, and see \cite{A09} for an introduction to
supergravity, while the reader is referred to \cite{WessBagger} for a thorough
exposition of both.} which further allows us to define a Batalin--Vilkovisky
(BV) extension of 3D supergravity where both the invariance under spacetime
diffeomorphisms and the supersymmetry are explicitly encoded in their
corresponding ghost fields. 

Such result might serve as the starting point for the analysis of the boundary
structure of this theory through the BV--BFV formalism \cite{Cattaneo2014}
developed by Cattaneo, Mnev and Reshetikhin, hence joining a broader effort to
better understand how gravity coupled to different kinds of matter fields should
be described, as well as how to quantize such theories, even in the presence of
non-trivial boundaries. Recent papers in this line include
\cite{MR3908845,cattaneo2024gravitycoupledscalarsun, Cattaneo_2024, canepa2022boundarystructuregaugematter, Canepa_2024}, to 
name a few. A further development of the results of this paper would be its extension 
to all lower-dimensional strata, as was done for PC theory in \cite{MR4557605}, and the 
construction of its space of quantum states, which for Einstein--Hilbert theory was done in
\cite{canepa2024doublebfvquantisation3d}, leveraging the relation with PC theory.

Concerning the structure of the document, the second section consists of a short
summary of the BV formalism \cite{Batalin1981} and of the AKSZ procedure
\cite{A06}. The third section, in turn, presents BV and $BF$ gravity, and leads
to the fourth and last section, where we present supergravity and build its BV
extension. Besides these sections, the reader can find an appendix where we will
compile a series of results of graded geometry that we shall use in the rest of
the article. Sources reviewing graded geometry are the review \cite{A02} or the
books \cite{Enn2019, DeWitt}. For an approach to three-dimensional supergravity 
in the context of graded geometry, see also \cite{Castellani:2016ibp,Castellani:2020kmz}.
Lastly, we let
the reader know that this article is an improved adaptation of the
content in the master thesis of the third author \cite{ASF}.

\subsection{Notation}
\begin{itemize}
\item $\N$ denotes the set of natural numbers \emph{including} $0$.
\item $\intrange{m,n} = \{k\in\mathbb Z\ |\ m \leq k \leq n\}$ denotes a range
of integers.
\item Given a Lie group $G$, its associated Lie algebra is denoted by $\mathfrak
g$. 
\item All forms of products are left implicit and deduced from the context,
unless some ambiguity is present, e.g.\  $\lambda\cdot u\otimes v =: \lambda u
v$.
\item In the absence of parentheses, a derivation $D$ acts only on the element
directly adjacent to it: $aDbc \coloneqq aD(b)c$.
\item Einstein's summation convention for pairs of upper and lower indices is
generally assumed: $x^i y^i \neq x^i y_i \coloneq \sum_{i=j} x^i y_j$.
\item Indices with a bar on them are not summed over: $x\nsu \imath y\nsd \imath
\neq x^i y_i$.
\item $\delta^i_j$ denotes Kronecker's delta: $\delta^i_j = 1$ for $i=j$,
$\delta^i_j = 0$ otherwise.
\item In the context of gravity or special relativity, Greek letters
designate spacetime indices while Latin indices designate Lorentz bundle
indices.
\item The equivalence sign ``$\equiv$'' is used to designate equality \emph{on
shell}.
\item There are multiple fields associated to a field $\psi$:
   \begin{itemize}
      \item $\cc \psi$ is it complex conjugate,
      \item $\hc \psi$ is its Hermitian conjugate,
      \item $\dc \psi$ is its Dirac conjugate,
      \item $\anti \psi$ is its associated antifield.
   \end{itemize}
\end{itemize}

\subsection*{Acknowledgements} We thank G.~Canepa, P.~Grassi and M.~Schiavina for their useful comments. We also thank the anonymous referees for helpful comments. 

\section{Classical BV formalism}

\noindent
The path integral formalism for quantum field theory relies on the possibility
of integrating out the quadratic terms in the Lagrangian density defining the
action, which is achieved through a generalisation to field theory of the saddle
point method---known as \emph{Feynman--Laplace method}---%
requiring the critical points of the action to form a finite subset of its
support. However, precisely because of the ``continuity''---as opposed to
``discreteness''---of topological groups, in theories described by a Lagrangian
with gauge freedom one can smoothly deform a critical point into another
critical point, resulting in critical loci that are themselves submanifolds.
This spoils the applicability of the aforementioned method, making manifest the
need for machinery that selects discrete subsets of the critical locus of an
action. This is precisely the issue that both the BRST and the BV formalisms
address. 

We choose to employ the latter because it has a greater range of applicability
than the former, and also because of its relation to the
Batalin--Fradkin--Vilkovisky (BFV) formalism, which allows for a perturbative
quantisation of field theories on the possible boundary of a manifold. Despite
all this, though, we will not be concerned with any quantisation in this
document, so our treatment of the theory will be minimal; an exhaustive approach
is found e.g.\  in \cite{Mnev}, and a review of the BFV formalism is e.g.\  in
\cite{A02}. Let us then summarise the BV formalism.

\begin{definition}
Given a graded manifold $\M$, a \term{cohomological vector field} $Q$
\index{Cohomological!vector field} is an element $Q\in\mathfrak X(\M)$ such
that
   \begin{align}\label{eq:cohomologicalVector}
   &Q^2 = 0, &
   &\abs Q = 1, &
   &\gh Q = 1,
   \end{align}
where one understands $Q^2$ as $Q\circ Q$, being a map
$\smooth{\M}\to\smooth{\M}$. A manifold endowed with such a vector field is a
\term{dg manifold}.\footnote{\ Here \emph{dg} stands for \emph{differential
graded}.}
\index{dg manifold}
\end{definition}

\begin{remark}
Given that $Q$ is odd, saying that $Q^2 = 0$ is equivalent to saying that $[Q,Q]
= 0$. 
\end{remark}

\begin{example}
A paradigmatic example of dg manifold is given by the odd tangent bundle $T[1]M$
of any non-graded manifold $M$. If $(x^i)$ are coordinates on $M$ and
$(\theta^i)$ the coordinates of ghost number $1$ on its fibre, then a dg structure 
is given by 
   \begin{equation}
   Q = \theta^i\partdev{x^i}.
   \end{equation}
\end{example}

\begin{definition}
Given a vector field $X\in\mathfrak{X}(\M)$ and a form \mbox{$\omega \in
\Omega(\M)$}, we say that $\omega$ is \term{X-invariant} \index{Invariant form}
if
   \begin{equation}
   \ld_X \omega = 0.
   \end{equation}
\end{definition}

\begin{definition}
A \term{dg-symplectic manifold} \index{dg-symplectic manifold} is a graded 
symplectic manifold
with a cohomological vector field $Q$ under which the symplectic form $\omega$
is invariant. That is, 
   \begin{equation}
   \ld_Q \omega = 0.
   \end{equation}
\end{definition}

\begin{definition} \index{Hamiltonian!vector field} Given a graded symplectic
manifold $\M$ and a function $f\in\smooth{\M}$, the \term{Hamiltonian vector
field} associated to $f$ is the field $\rvec f \in \mathfrak X(\M)$ such that
   \begin{equation}
   \rvec f = \sign{\abs f + 1}\{f,\bullet\}.
   \end{equation}
\end{definition}

\begin{definition}
A \term{dg-Hamiltonian manifold of degree $k$} \index{dg-Hamiltonian!manifold}
$(\M, H, Q, \omega)$ is a dg-symplectic manifold $\M$ where the symplectic form
$\omega$ has ghost number $k$ and the cohomological vector field $Q$ is given as
the Hamiltonian vector field of some degree $k+1$ function $H$, called its
\term{Hamiltonian function}\index{Hamiltonian!function}:
   \begin{equation}
   Q = \sign{k} \{H,\bullet\}.
   \end{equation}
\end{definition}

\begin{definition}
A \term{classical BV theory} \index{Classical BV theory} is a dg-Hamiltonian
manifold of degree $-1$, that is, a tuple $(\mathcal F, S, Q, \omega)$ such that
   \begin{enumerate}
   \item $\mathcal F$ is a graded symplectic manifold, called the
   \term{field space}.\index{BV!field space}
   \item $S$ is an even function over $\mathcal F$ of degree $0$,
   called the \term{BV action}\index{BV!action}.
   \item $Q$ is the cohomological, Hamiltonian vector field of $S$.
   \item $\omega$ is a $Q$-invariant, odd symplectic form over $\mathcal F$ of
   degree $-1$.
   \end{enumerate}
Given that $\mathcal F$ is odd symplectic, we can locally associate each field
$\phi\in\mathcal F$ to another field $\anti\phi\in\mathcal F$, known as its
\index{Antifield} \term{antifield}, that by construction satisfies
   \begin{align}
   &\dego \anti\phi = n - \dego \phi, 
   &\gh \anti\phi = -\gh \phi - 1.
   \end{align}
\end{definition}

\begin{definition}
A BV theory $(\mathcal F, S, Q,\omega)$ is a \term{BV extension}
\index{BV!extension} of a classical theory described by a space of fields
$\cl{\mathcal F}$ and an action $\cl S$, if the ghost number zero part of
$\mathcal F$ and of $S$ correspond to $\cl{\mathcal F}$ and to $\cl S$, and if
the restriction $\left. Q\right|_{\cl{\mathcal F}}$ yields the gauge symmetries. 
\end{definition}

\begin{definition}
Two BV theories are \term{weakly equivalent} \index{Weak equivalence} if both 
of them are BV extensions of a same classical theory.\par
These theories will be \term{strongly equivalent} \index{Strong equivalence} if
there is a graded symplectomorphism $\Phi:\mathcal F\to\mathcal F'$ between
their respective field spaces that pulls the action of one theory back to the
action of the other:
   \begin{equation}
   \phi^*S' = S.
   \end{equation}
Such a symplectomorphism is known as a \index{Canonical!transformation}
\term{canonical transformation}.
\end{definition}
   
\begin{definition}
Let $(q,p)$ and $(q',p')$ be the respective even-odd coordinates of two
different BV field spaces $\mathcal F$ and $\mathcal F'$. A \index{Generating
function} \term{graded generating function of type $j$}, for
$j\in\intrange{1,4}$, is a graded function $G_j$ that we use to define two
coordinates among $q,p,q',p'$ as a function of the remaining two, in one of the
four following ways: \par 
   \begin{subequations} \label{eq:genFunctions}
   \begin{align}
   p &= \sign{\abs q+1} \frac{\partial G_1(q,q')}{\partial q}, &
   p' &= \sign{\abs {q'}} \frac{\partial G_1(q,q')}{\partial q'}, \\
   p &= \sign{\abs q+1} \frac{\partial G_2(q,p')}{\partial q}, &
   q' &= \sign{\abs {p'}} \frac{\partial G_2(q,p')}{\partial p'}, \\
   q &= \sign{\abs p+1} \frac{\partial G_3(p,q')}{\partial p}, &
   p' &= \sign{\abs {q'}} \frac{\partial G_3(p,q')}{\partial q'}, \\
   q &= \sign{\abs p+1} \frac{\partial G_4(p,p')}{\partial p}, &
   q' &= \sign{\abs {p'}} \frac{\partial G_4(p,p')}{\partial p'}.
   \end{align}
   \end{subequations}
\end{definition}

\begin{remark}
By design, generating functions have top cohomological degree and ghost number
$-1$.
\end{remark}

\begin{definition}
\label{def:mappingSpace} Given two graded manifolds $\M$, $\MN$, and letting
$\op{Mor}(\M, \mathcal N)$ be the manifold of grade-preserving morphisms
$\M\to\mathcal N$ in the category of graded manifolds,  the \term{mapping space}
\index{Mapping space} $\op{Map}(\M, \mathcal N)$ is the extension of
$\op{Mor}(\M, \mathcal N)$ that includes grade-shifting maps.
\end{definition}

\begin{remark}
If $\mathcal N$ is a graded vector space then
   \begin{equation}
   \op{Map}(\M, \mathcal N) = \smooth{\M}\otimes\mathcal N,
   \end{equation}
so locally the mapping space will have the form of such a tensor product of
graded spaces. Details of this definition can be found in \cite{Cattaneo2014}.
\end{remark}

\begin{definition}
An \term{AKSZ theory}\footnote{\ Here we focus on this special case of the more
general AKSZ method.} \term{in $n$ dimensions} \index{AKSZ!theory} $(M, \mathcal
N, H, Q, \alpha)$ is the combination of two things:
   \begin{enumerate}
   \item A \term{source} \index{AKSZ!source} consisting of a closed and oriented
   $n$-manifold $M$.
   \item A \term{target} \index{AKSZ!target} consisting of a dg-Hamiltonian
   manifold $(\mathcal N, H, Q, \omega)$ of degree $n-1$ whose symplectic form
   $\omega$ is exact: 
      \begin{equation}
      \omega = d_\MN \alpha 
      \end{equation}
   for $d_\MN$ the exterior derivative on $\mathcal N$. 
   \end{enumerate}
\end{definition}

\begin{definition}
Given an AKSZ theory $(M, \MN, H_\MN, Q_\MN, \alpha_\MN)$ in $n$ dimensions, we
define the \term{AKSZ fields space} \index{AKSZ!fields space}
   \begin{equation}
   \mathcal F = \op{Map}(T[1] M, \mathcal N).
   \end{equation}
Employing the notation in Remark \ref{rem:notationAKSZ}, we take the
evaluation map $\op{ev} : T[1] M \times \mathcal F \to \mathcal N$ and define,
for all $k \in \N$ and for coordinates $\xi$ on $\MN$ and $X$ on $\mathcal F$,
its pullback as
   \begin{equation}
   \op {ev}^* : \Omega^k(\MN) \to \Omega(M) \otimes \Omega^k(\mathcal F):
   \beta(\xi) \mapsto \widehat \beta(X),
   \end{equation}
and we define  the pushed forward projection
   \begin{equation}
   \begin{aligned}
   \pi_*:&\ \Omega(M) \otimes \Omega^k(\mathcal F)
   \ \to\ \Omega^k(\mathcal F), \\
   &\ \varphi \otimes \Phi\ \mapsto \int_M \varphi^{\text{top}} \otimes \Phi
   \quad \forall\ \varphi \in \Omega(M),\ \Phi \in \Omega^k(\mathcal F).
   \end{aligned}
   \end{equation}
\index{AKSZ!superfields}We construct coordinates $(X^i)$ on $\mathcal F$---the
so called \term{AKSZ super\-fields}---associated to the coordinates $(x^i)$ on
$\MN$ as
   \begin{equation}
   X^i = \op {ev}^* x^i,
   \end{equation}
and compose $\pi_*$ with $\op {ev}^*$ to produce the \term{transgression map}
\index{Transgression map} 
   \begin{equation}
   \mathcal T =\pi_*\op{ev}^*.
   \end{equation}
Letting $d_M$ be the exterior derivative on $M$, and further letting $\widetilde
d_M, \widetilde Q_\MN\in\mathfrak X(\mathcal F)$ be the respective lifts to
$\mathcal F$ of $d_M$ and of $Q_\MN$, we finally define the \term{AKSZ
construction} \index{AKSZ!construction} associated to this theory as the tuple
$(\mathcal F, S, Q, \omega)$ for
   \begin{enumerate}
   \item the \term{AKSZ action} 
   \begin{equation}
   S = \iota_{\widetilde d_M} \mathcal T \alpha_\MN + \mathcal T H_\MN,
   \end{equation}
   \item the \term{AKSZ vector field} \index{AKSZ!vector field} 
   \begin{equation}
   Q = \widetilde d_M + \widetilde Q_\MN,
   \end{equation}
   \item and the \term{AKSZ symplectic form} 
   \begin{equation}
   \omega = \mathcal T\omega_\MN.
   \end{equation}
   \end{enumerate}
\end{definition}

\begin{remark} \label{rem:notationAKSZ} Given a form $\beta \in
\Omega(\MN)$, and coordinates $x$ over $\MN$ and $X$ over $\mathcal F$, by
$\widehat \beta(X)$ one understands the coordinates expression for $\beta(x)$, but
symbolically replacing $x$ by $X$. \par%
Given any manifold $M$ and form $\varphi \in \Omega(M)$, $\varphi^{\text{top}}$
denotes the components of $\varphi$ with top cohomological degree, that is,
those components such that $\deg_{\Omega(M)} \varphi^{\text{top}} = \dim M$.
\end{remark}

\begin{remark}
The maps $\pi_*$ and $\mathcal T$ are graded: 
   \begin{equation}
   \abs{\pi_*} = \abs{\mathcal T} = -\dim M.
   \end{equation}
\end{remark}

\begin{remark} \label{rem:localAKSZ}
By Definition \ref{def:mappingSpace}, locally  
   \begin{equation}
   \mathcal F \cong \Omega(M)\otimes \MN,
   \end{equation}
so in practice we can write $S$, $Q$ and $\omega$ explicitly:
   \begin{align}
   S &= \int_M \Bigl( (\widehat \alpha_\MN)_i(X)\: d_M X^i + \widehat H_\MN(X) \Bigr), \\
   Q &= \int_M \Bigl( d_M X^i + \widehat Q_\MN^i(X) \Bigr)\ \partdev{X^i}, \\
   \omega &= \sign{n}\int_M (\widehat \omega_\MN)_{ij}(X)\: d_\MN X^i \wedge d_\MN X^j.
   \end{align}
Given this, together with the fact that each superfield $X^i$ can be decomposed
over summands $\{X_{(j)}^i\}_j$ of definite cohomological degree
$j\in\intrange{0,n}$, at the end of the day the AKSZ action is the action that
we would obtain by 
   \begin{enumerate}
   \item symbolically replacing the original fields with their associated
   super\-fields,%
   \item expanding those in components of definite cohomological degree,
   \item keeping only the terms which have top cohomological degree.
   \end{enumerate}
An example of such construction is given in the next section, so let us proceed
without further ado.
\end{remark}


\section{3D gravity in vacuum and $BF$ theory}

\begin{definition}
An $n$-dimensional \index{Spacetime} \term{spacetime} is a closed
$n$-manifold with a \index{Mostly positive metric} \term{mostly positive}
metric, that is, of signature $(p = n-1, n = 1)$.
\end{definition}

\begin{definition}
Given a principal $\op{SO}(n-1,1)$-bundle $P$ over an $n$-dimensional manifold,
its \index{Mikowski!bundle} \term{Minkowski bundle} is the associated vector
bundle $(\mathcal V, \eta)$ with typical fibre $\R^n$, endowed with the
\index{Minkowski!metric} \term{Minkowski metric}
   \begin{equation}
   \eta \coloneq (-\mathbb I_{1}) \oplus \mathbb I_{(n-1)},%
   \end{equation}
for $\mathbb I_{k}$ the identity matrix in $k$ dimensions. 
\end{definition}

\begin{remark}
Here we will focus on the case where $n=3$, so 
   \begin{equation}
   \eta = \op{diag}(-1,1,1).
   \end{equation}
\end{remark}

\begin{definition} \index{Cotriad} \index{Triad} \index{Coframe field}
\index{Frame field} Given a Minkowski bundle $(\mathcal V,\eta)$, a
\term{cotriad} or \term{coframe field} over a $3$-spacetime $M$ is a
non-degenerate $1$-form $e$ over $M$ valued in $\mathcal V^{\wedge 1}$.
Associated to it there is a \term{triad} or \term{frame field} $\inv e$, which
is its inverse in the following sense:
   \begin{align}
   & e \in \Omega^1(M, \mathcal V^{\wedge 1}), &
   & \inv e \in \Omega^1(M, \mathcal V^{\wedge 1})^*, &
   & \inv e(e) = 1.
   \end{align}
\end{definition}

\begin{remark}
We talk of $\mathcal V^{\wedge 1}$ and not simply of $\mathcal V$ because we
associate to the (co-)triads a multivector degree
   \begin{equation}
   \deg_{\mathcal V} e = - \deg_{\mathcal V} \inv e = 1.
   \end{equation}
From now on, we will always conceive $\mathcal V$ as $\mathcal V^{\wedge 1}$.
\end{remark}

\begin{definition}
The \term{Palatini--Cartan formalism} in three dimensions, \index{PC
gravity} or simply \index{3D!gravity} \term{3D gravity}, consists of
   \begin{enumerate}
   \item an oriented spacetime $M$, assumed to have no boundary,
   \item a principal $\op{SO}(2,1)$-bundle $P \to M$, with Minkowski bundle
   $\mathcal V$,
   \item a possibly non-zero \index{Cosmological constant} \term{cosmological
   constant} $\Lambda\in\R$,
   \end{enumerate}
together with a field content given by
   \begin{enumerate}
   \item[4.] a {cotriad} $e \in \Omega^1(M, \mathcal V)$ over $M$,
   \item[5.] a connection $1$-form $\Gamma\in\op{Conn}(P) \cong \Omega^1 (M, \mathcal
   V^{\wedge 2})$ over $P$.
   \end{enumerate}
With these we define the classical field space
   \begin{equation}
   \fgr^0 = \Omega^1_{\textup{nd}}(M, \mathcal V)\times \op{Conn}(P)\ 
   \footnote{\ Where ``nd'' stands for \emph{non-degenerate}.}
   \end{equation}
and the classical action
   \begin{equation}
   \sgr^0(\Lambda) = 
   \int_M \left\langle e \wedge F_{\Gamma} 
   + \frac \Lambda 6 e^{\wedge 3} \right\rangle,
   \end{equation}
where the angle brackets $\langle\bullet\rangle$ designate the appropriate
contraction of any indices other than those over $\Omega(M)$.%
\footnote{\ The expression within the brackets actually defines a density, which
could be integrated over a non-orientable manifold, yet for simplicity we
assume orientability.}
\end{definition}

\begin{remark}
The reason why we can say that connections take values in $\mathcal V^{\wedge
2}$ follows from the fact that, since we work in three dimensions, $V^{\wedge 2}
\cong \mathfrak{so}(2,1)$ as a Lie algebra.  
\end{remark}

\begin{remark}
For convenience and when it is possible, we often express a particular action
$S_i$ as the integral of a top form $\langle \mathscr L_i \rangle$:
   \begin{equation}
   S_i = \int_M \Bigl\langle \mathscr L_i \Bigr\rangle. 
   \end{equation}
\end{remark}

\begin{remark}
By \emph{appropriate contractions} one understands that ultimately $\sgr^0$ must
be a scalar, so any internal index must be contracted. This is to say that we
must correctly ``trace over'' the multivector indices that both $e$ and
$F_\Gamma$ have over $\mathcal V$, leading specifically to
   \begin{align}
   &\bigl\langle e \wedge F_\Gamma \bigr\rangle 
   = \epsilon_{abc}\ e^a \wedge F_\Gamma^{bc},
   &\bigl\langle e^{\wedge 3} \bigr\rangle 
   = \epsilon_{abc}\ e^a \wedge e^b \wedge e^c,
   \end{align}
for the $3$-dimensional Levi--Civita symbol $\epsilon_{abc}$ and leaving implicit
the decomposition over the generators of $\Omega (M)$.
\end{remark}

\begin{definition}
Given a connection $1$-form $\Gamma$ and its associated covariant derivative
$D_\Gamma$, the \index{Lie!covariant derivative} \term{Lie covariant
derivative} $\lxi^\Gamma$ with respect to a vector field $\xi$ is defined as
   \begin{equation}
   \lxi^\Gamma = [\ixi, D_\Gamma].
   \end{equation}
\end{definition}

\begin{construction}[BV extension of 3D gravity] \label{construct:3DGravity}
Based on the data from 3D gravity, let $\fgr$ be the field space
   \begin{equation}
   \fgr = T^*[-1]\bigl( \fgr^0 \times \mathfrak X(M)[1] 
   \times \Omega^0(M, \op{ad}(P)[1]) \bigr),
   \end{equation}
where $\mathfrak X(M)[1]$ is the space of vector fields over $M$ with shifted
degree and $\op{ad}(P)[1]$ stands for the degree-shifted adjoint bundle of $P$.
On the base of this shifted cotangent bundle we denote the fields as
   \begin{equation}
   \bigl((e, \Gamma), \xi, \chi \bigr) 
   \in \fgr^0 \times \mathfrak X(M)[1] 
   \times \Omega^0\bigl(M, \op{ad}(P)[1]\bigr),
   \end{equation}
where $\xi$ encodes the diffeomorphism invariance of general relativity and
$\chi$ is the ghost field required for the internal gauge transformations. In
turn, the fibre includes their associated antifields, and the fact that $\fgr$
is a shifted cotangent bundle allows us to define the canonical symplectic form
   \begin{equation}
   \omgr = \int_M \Bigl( 
   \delta  e \delta \anti e 
   + \delta  \Gamma \delta \anti \Gamma 
   + \delta  \xi \delta \anti \xi 
   + \delta  \chi \delta \anti \chi 
   \Bigr)
   \end{equation}
Finally, we can define the action $\sgr(\Lambda) \coloneq \sgr^0(\Lambda) + \sgr^1$
with
   \begin{multline}
   \sgr^1 = \int_M \Bigl\langle 
      \anti e \bigl(\lxi^\Gamma e + [\chi, e]\bigr)
      + \anti\Gamma \bigl(D_\Gamma \chi + \ixi F_\Gamma\bigr) \\
      + \frac 12 \iota_{[\xi,\xi]}\anti\xi
      + \frac 12 \anti \chi \bigl([\chi, \chi] - \ixi^2 F_\Gamma \bigr)
   \Bigr\rangle,
   \end{multline}
and its associated cohomological vector field $\qgr$, which acts as
   \begin{subequations}
   \begin{align}
   \qgr (e) &= \lxi^\Gamma e + [\chi, e], \\
   \qgr (\Gamma) &=  D_\Gamma \chi + \ixi F_\Gamma, \\
   \qgr (\xi) &= \frac 12 [\xi, \xi] \\
   \qgr (\chi) &= \frac 12 \bigl([\chi, \chi] - \ixi^2 F_\Gamma\bigr).
   \end{align}
   \end{subequations}
We then write $\tgr \coloneq (\fgr, \sgr, \qgr, \omgr)$ to refer to this theory.
\end{construction}

\begin{proposition}
The tuple $\tgr$ is a BV extension of 3D gravity, which we shall call
$3$-dimensional \index{BV!gravity} \term{BV gravity}.
\end{proposition}

\begin{definition}
A \index{$BF$!theory} \term{$BF$ theory} in $n\geq 2$ dimensions is defined as a
pair $(\mathcal F, S)$ where, given
   \begin{enumerate}
   \item an $n$-dimensional oriented manifold $M$,
   \item a finite dimensional Lie group $G$ and a $G$-bundle $P$ over $M$,
   \item a form $B\in\Omega^{n-2}(M, \op{ad}^*(P))$ valued in the coadjoint
   bundle,
   \item a connection form $A\in\op{Conn}(P)$,
   \end{enumerate}
one sets the field space to be
   \begin{equation}
   \mathcal F = \op{ad}^*(P)\times \op{Conn}(P)
   \end{equation}   
and the action to be
   \begin{equation}
   S = \int_M \langle B, F_A\rangle,
   \end{equation}
where $F_A$ is the curvature of $A$ and $\langle\bullet,\bullet\rangle$ is the
pairing of dual maps.
\end{definition}

\begin{definition}
We define \index{$BF$!gravity} \term{$BF$ gravity} as the $BF$ theory $(\fbf^0,
\sbf^0)$ where $n=3$, and where $P$ is the $\op{SO}(2,1)$-bundle over $M$ with
associated Minkowsky bundle $\mathcal V$ for some reference Lorentzian metric.

Moreover, for some reference connection $A'$ we interpret $B$ and $A-A'$ as
$1$-forms valued in $\mathcal V^*$ and $\mathcal V^{\wedge 2}$ respectively, and
further use the internal Minkowski metric to identify $\mathcal V^*$ with
$\mathcal V$, hence seeing $B$ as taking values in the latter. Thus the field
space becomes 
   \begin{equation}
   \fbf^0 \cong \Omega^1(M, \mathcal V)
   \times \Omega^1(M,\mathcal V^{\wedge 2}).
   \end{equation}
Finally, we incorporate a cosmological term in the action, resulting in
   \begin{equation}
   \sbf^0(\Lambda) = \int_M \bigl\langle BF_A 
   + \frac \Lambda 6 B^{\wedge 3}\bigr\rangle,
   \end{equation}
where $\langle \bullet \rangle$ denotes again the trace over the internal
indices.
\end{definition}

\begin{proposition}[\cite{A01}]\label{prop:3DBFgravity} $BF$ gravity is an AKSZ
theory and its AKSZ extension $\tbf\coloneq(\fbf, \sbf, \qbf, \ombf)$ is strongly
equivalent to $3$-dimensional BV gravity. A canonical transformation
   $\phibf: \sgr\to\sbf$
is provided by the following type $2$ generating function:
   \begin{multline}
   \label{eq:genFct0}
   G(q, p')_{\text{$BF$}}  =  
   - \anti B (e - \ixi \anti \Gamma -\frac 12 \ixi^2 \anti \chi)
   - A \anti \Gamma \\ 
   - \anti \tau(-\ixi e + \frac 12 \ixi^2 \anti \Gamma 
      - \frac 13 \ixi^3 \anti \chi)
   - c \anti \chi
   \end{multline}
for $q \coloneq (e, \anti \Gamma, \xi, \anti \chi)$
and $p' \coloneq (\anti B, A, \anti \tau, c)$.
\end{proposition}

\begin{samepage}
\begin{remark}
We recall the decomposition of the fields that results from this proposition,
keeping $\anti \tau$ implicit to avoid cramping the equations:
   \begin{subequations} \label{eq:BFto3Dfull}
   \begin{align}
   B &= e - \ixi \anti \Gamma -\frac 12 \ixi^2 \anti \chi, & 
   \anti B &= \anti e - \ixi \anti \tau, \\
   A &= \Gamma - \ixi \anti e + \frac 12 \ixi^2\anti \tau, &
   \anti A &= \anti \Gamma, \label{eq:BFto3Dfull:2}\\
   \tau 
   &= - \ixi e + \frac 12 \ixi^2 \anti \Gamma + \frac 13 \ixi^3 \anti \chi, &
   \anti \tau 
   &= \inv e(\anti \xi - \anti e \anti \Gamma + \ixi \anti e \anti \chi), \\
   c &= \chi + \frac 12 \ixi^2 \anti e - \frac 16 \ixi^3 \anti \tau, &
   \anti c &= \anti \chi.
   \label{eq:BFto3Dfull:4}
   \end{align}
   \end{subequations}
\end{remark}
\end{samepage}

\section{BV supergravity}

\begin{definition} \label{def:supergravity}
We define \index{3D!supergravity}\term{3D supergravity} and
\index{$BF$!supergravity} \term{$BF$ supergravity} by extending respectively
$\fgr^0$ and $\fbf^0$ to
   \begin{align}
   &\fgrf^0 = \fgr^0\times\Pi \mathcal S(P) & &\hbox{and} & &\fbff^0 = \fbf^0\times \Pi \mathcal S(P),
   \end{align}
for $\mathcal S(P)$ the space of $1$-forms taking values in the spinor bundle associated
to the principal bundle $P$, and then extending their actions by a
\index{Rarita--Schwinger action} \term{Rarita--Schwinger action} term
   \begin{equation}
   \begin{aligned}
   \sgrf^0(\Lambda) &= \sgr^0(\Lambda) + \int_M \frac 12 \dc \psi\, D_\Gamma \psi, \\
   \sbff^0(\Lambda) &= \sbf^0(\Lambda) + \int_M \frac 12 \dc \varphi\, D_A \varphi.
   \end{aligned}
   \end{equation}
Here the fields $\psi,\varphi \in \Pi \mathcal S(P)$ are spin $\frac 32$ Majorana
spinors.
\end{definition}

\begin{remark}
The connection forms will act on the spinor fields through the spin $\frac 32$
real (Majorana) representation of the algebra $\mathfrak{spin}(2,1)$.
This representation will be
denoted by $\bm \rho$ and, for the Pauli matrices $\{\sigma^a\}_{a=1}^3$, we will
use the shorthand notation $\rho^a \coloneq \bm \rho(\sigma^a)$. Letting $\{v_a\}$
be a basis of the sections of the Minkowski bundle $\mathcal V$, we write
   \begin{equation}
   \rho \coloneq \rho^a v_a,
   \end{equation}   
with which the equations of motion that follow from $\sbff^0$ take the following form:
   \begin{subequations} \label{eq:eqsMotion}
   \begin{align}
   F_A + \frac \Lambda 2 B \wedge B &= 0,\\
   D_A \varphi &= 0, \\
   D_A B + \frac 12 \dc{\varphi}\, \rho \varphi &= 0. \label{eq:eqsMotion:3}
   \end{align}
   \end{subequations}
Meanwhile, the equations of motion issued from $\sgrf^0$ are analogous, after
replacing
   \begin{equation}\label{eq:BFto3D}
   (B, A, \varphi)\ \leftrightarrow\ (e, \Gamma, \psi).
   \end{equation}
\end{remark}

\begin{proposition} \label{prop:AKSZsupergravity}
$BF$ supergravity is an AKSZ theory.
\end{proposition}

\begin{proof}
We take the spacetime $M$ as the source manifold, let $V$ be the Minkowski space
that is the typical fibre of the Minkowvski bundle $\mathcal V$, and $S(P)$ be the
vector space associated to the spinor representation $\bm \rho$. We hence define
$(b,a,f)$ as coordinates on the target
   \begin{equation}
   \MN 
   \cong V[1] \oplus V^{\wedge 2}[1] \oplus \Pi S(P)[1].
   \end{equation}
We see from this that $\MN$ is endowed with the symplectic form
   \begin{align}
   \omega_\MN &= d_\MN \alpha_\MN, \label{eq:omegaTarget} \\
   \Omega^1(\MN)\ni \alpha_\MN &= b\,d_\MN a + \frac 12 \dc f\, d_\MN  f,
   \end{align}
the Hamiltonian function
   \begin{equation}
   H_\MN= \Bigl\langle \frac 12 b[a,a] + \frac \Lambda 6 b^3 \Bigr\rangle
   + \frac 12 \dc f\, a  f,
   \end{equation}
and the cohomological vector field $Q_\MN = -\{H_\MN, \bullet\}$ associated to
$H_\MN$ through the Poisson bracket induced by $\omega_\MN$. 
\end{proof}

\begin{construction}[AKSZ extension of $BF$ supergravity]
\label{construct:AKSZsupergravity} Given the target established
in the proof of Proposition \ref{prop:AKSZsupergravity}, the field space is
   \begin{equation}
   \fbff \cong \Omega\bigl(M,\mathcal V\bigr)[1] 
   \oplus \Omega\bigl(M,\mathcal V^{\wedge 2}\bigr)[1] 
   \oplus \Omega\bigl(M, \Pi \mathcal S(P)\bigr)[1].
   \end{equation}
On it, the coordinates are given by the superfields
   \begin{subequations}
   \begin{align}
   \widetilde b &=  \tau + B + \anti A + \anti c, \\
   \widetilde a &=  c + A + \anti B + \anti \tau, \\
   \widetilde  f &=  \gamma + \varphi + \anti \varphi + \anti \gamma,
   \end{align}
   \end{subequations}
where summands are ordered in increasing cohomological degree, from $0$ to $3$,
and decreasing ghost number, from $1$ to $-2$. Besides, note that the parity of
a field differs by $1$ from the parity of its corresponding antifield.
Now, replacing the classical fields in $\sbff^0$ by their associated superfield,
keeping only those terms of cohomological degree $3$ and rearranging them, we
find the BV action for $BF$ supergravity:
   \begin{subequations}\label{eq:SBFF}
   \begin{equation}
   \sbff(\Lambda) = \int_M \Bigl\langle 
      \lbff^0(\Lambda) + \lbff^1(\Lambda) + \lbff^2,
   \Bigr\rangle,
   \end{equation}
where
   \begin{align}
   \lbff^0(\Lambda) &=
      BF_A + \frac \Lambda 6 B^{\wedge 3} 
      + \frac 12 \dc \varphi\, D_A \varphi,
   \\
   \lbff^1(\Lambda) &= 
      \anti B \bigl( [c,B] + D_A\tau \bigr) 
      + \anti A \bigl( D_A c + \Lambda B \tau \bigr)  \\
      &\hspace{35mm}+ \frac 12 \anti c \bigl( [c,c] + \Lambda \tau \tau \bigr)
      + \anti \tau [c,\tau],\notag
   \\
   \lbff^2 &=
      \dc{\gamma}\, \anti B \varphi
      + \frac 12 \dc \gamma \, \anti \tau \gamma
      + \dc{\anti \varphi}\,\bigl( D_A\gamma + c \varphi \bigr)
      + \dc{\anti \gamma}\, c \gamma.
   \end{align}
   \end{subequations}
The symplectic form is given by \eqref{eq:omegaTarget} when we replace the
coordinates by their corresponding superfields and keep only the terms of
ghost number $-1$, resulting in
    \begin{multline}
    \omega_\MN = \int_M \Bigl\langle
    d_\MN B\; d_\MN \anti B 
    + d_\MN A\; d_\MN \anti A 
    + d_\MN \dc \varphi\; d_\MN \anti \varphi \\
    + d_\MN \tau\; d_\MN \anti \tau 
    + d_\MN c\; d_\MN \anti c 
    + d_\MN \dc \gamma\; d_\MN \anti \gamma
    \Bigr\rangle.
    \end{multline}
In turn, we can read off the action \eqref{eq:SBFF} the way in which the
cohomological vector field acts on the coordinate fields:
   \begin{subequations} \label{eq:QBFF}
   \begin{align}
   \qbff(B) & =  [c,B] + D_A\tau + \dc \gamma \, \rho \varphi, &
   \qbff(\tau) & =  [c,\tau] + \frac 12 \dc \gamma \, \rho \gamma, 
   \label{eq:QBFF:1} \\
   \qbff(A) & =  D_A c + \Lambda B \tau, &
   \qbff(c) & =  \frac 12 [c,c] + \frac 12 \Lambda \tau \tau, \\
   \qbff(\varphi) & =  D_A \gamma +  c\varphi, &
   \qbff(\gamma) & =  c\gamma \vphantom{\frac 12}.
   \end{align}
   \end{subequations}
Thus we conclude the construction collecting everything in a tuple 
   \begin{equation}
   \tbff \coloneq (\fbff, \sbff, \qbff, \ombff).
   \end{equation}
\end{construction}

\begin{lemma}
For any degree $1$ vector field $\xi\in\mathfrak X(M)[1]$, principal connection
$1$-form $\Gamma\in\op{Conn}(P)$ and field $\varphi\in\Omega(M,\mathcal V)$
valued in an associated vector bundle $\mathcal V$, the following identity
holds:
   \begin{equation}
   \bigl[\lxi^\Gamma, \ixi\bigr] \varphi = \iota_{[\xi,\xi]}\varphi.
   \end{equation}
\end{lemma}

\noindent This is proven in \cite[Lemma 9]{A12}.

\begin{proposition} \label{prop:onShellQGR}
There is a ghost fermion $\varepsilon$ such that on shell 
   \begin{equation}
   \label{eq:onShellQGR}
   \qgrf(\xi)\equiv \qgr(\xi) 
   + \frac 12 \dc \varepsilon\, \inv e(\rho ) \varepsilon.
   \end{equation}
\end{proposition}

\begin{proof}
Let us call $Q'\coloneq \qgrf - \qgr$ the on-shell extension of $\qgr$.
Now, on shell all antifields are set to zero, so \eqref{eq:BFto3Dfull} reduces
to
   \begin{equation}
   (B, A, \tau, c) = (e,\Gamma, -\ixi e, \chi),
   \end{equation}
which we extend additionally with $(\varphi, \gamma) =  (\psi, \kappa)$,
thus being able to translate on shell the first part of \eqref{eq:QBFF:1} to 
   \begin{equation}
   \begin{aligned} 
   \qgrf(e) &\equiv -\Dg(\ixi e) + [\chi, e] 
   + \dc \kappa \, \rho \psi \\
   & = -\ixi \Dg e + \lxig e + [\chi, e]
   + \dc \kappa \, \rho \psi. 
   \end{aligned}
   \end{equation}
Meanwhile, the fact that $\abs e = 2$, together with \eqref{eq:LIotaCommutator}
and the definition of $\lxig$, imply that
   \begin{equation}
   \ixi \bigl(\lxig e\bigr) = \frac 12 \bigl([\ixi^2, D_\Gamma] 
   - \iota_{[\xi,\xi]}\bigr) e.
   \end{equation}
Recalling that $\qgr(\xi) = \frac 12 [\xi, \xi]$, we use all this to further
find that 
   \begin{equation}
   \begin{aligned}
   \qgrf (\ixi e) &= [\qgrf, \ixi]e + \ixi (\qgrf e) 
   = \iota_{\qgrf(\xi)} e + \ixi(\qgrf e) \\
   &\equiv \iota_{\qgrf(\xi)} e
      -\ixi^2 \Dg e + \ixi\bigl(\lxig e\bigr) 
      + \ixi [\chi, e] + \dc \kappa \, \rho \ixi \psi \\
   & = \iota_{Q'(\xi)} e + \frac 12 \iota_{[\xi, \xi]} e - \ixi^2 \Dg e
      + \frac 12 \ixi^2 \Dg e - \frac 12 \Dg (\ixi^2 e) \\
      & \hspace{42mm}- \frac 12 \iota_{[\xi,\xi]} e
      + \ixi [\chi,e] + \dc {\kappa} \, \rho \ixi \psi \\
   & = \iota_{Q'(\xi)} e -\frac 12 \ixi^2 \Dg e 
   + \ixi [\chi, e] + \dc \kappa \, \rho \ixi \psi.\\
   \end{aligned}
   \end{equation}
Moreover, after adapting to on-shell 3D supergravity both the equation of motion
\eqref{eq:eqsMotion:3} and the second part of \eqref{eq:QBFF:1}, from what
precedes we deduce that 
   \begin{equation}
   \begin{aligned}
   \qgrf(\ixi e)
   &\equiv \iota_{Q'(\xi)} e + \frac 12 \ixi \dc \psi \, \rho \ixi \psi
   + \ixi [\chi, e] + \dc \kappa \, \rho \ixi \psi \\
   &\equiv -\qgrf(\tau) 
   \equiv \ixi [\chi, e] -\frac 12 \dc \kappa \, \rho \kappa,
   \end{aligned}
   \end{equation}
which holds if and only if
   \begin{equation}
   \iota_{Q'(\xi)} e \equiv -\frac 12 (\dc \kappa + \ixi \dc \psi)\, \rho
   (\kappa + \ixi \psi).
   \end{equation}
With this we finally conclude that 
   \begin{equation}
   \begin{aligned}
   \qgrf(\xi) 
   &= \qgr(\xi) + \inv e\bigl(e(Q'\xi)\bigr)
   = \qgr(\xi) - \inv e\bigl(\iota_{Q' (\xi)} e\bigr) \\
   &\equiv \qgr(\xi) + \frac 12 \dc \varepsilon \, \inv e(\rho) \varepsilon
   \end{aligned}
   \end{equation}
for the ghost Majorana fermion $\varepsilon \coloneq \kappa + \ixi \psi$.
\end{proof}

\begin{remark}
This property is expected, since the generators of supersymmetry square to
translation generators, and the former are to be encoded by ghost fermions while
the latter are realised through $\xi$.
\end{remark}

\begin{remark}
Our current goal being to extend \eqref{eq:genFct0} as to find a BV theory of 3D
supergravity that encodes explicitly both supersymmetry and diffeomorphism
invariance, this last proposition will serve us as guiding principle. Indeed, we
will be searching for a type 2 generating function $G_{\text{$BF$\textPhi}}$ that
decomposes as
   \begin{equation}
   G_{\text{$BF$\textPhi}} = G_{\text{$BF$}} + \egbf,
   \end{equation}
and evidently we would like $\ext[$BF$] G$ to be  a minimal extension, that is, as
simple as possible without being ineffective. This \emph{without being
ineffective} is precisely what Proposition \ref{prop:onShellQGR} addresses:
we must ensure that the extended symplectomorphism $\phibff: \fgrf\to\fbff$
leads to a cohomological vector field that on shell is equal to
\eqref{eq:onShellQGR}. Fortunately, finding such extension is eased by the
next proposition.
\end{remark}

\begin{proposition}\label{prop:extensionIsEasy}
A minimal extension $\egbf$ of $\gbf$ ensuring that equation
\eqref{eq:onShellQGR} holds on shell can only depend on spinorial fields or on
contractions of those with respect to $\xi$.
\end{proposition}

\begin{proof}
As before, we denote by $(\varphi, \gamma)$ the spinorial field and ghost in
$\fbff$ and by $(\psi, \varepsilon)$ the corresponding pair on $\fgrf$. Since the
extension we are looking for aims at being minimal, all its terms must be
spinorial scalars, because if any term in $\egbf$ did not include spinors, it
would effectively amount to a modification of $\gbf$ spoiling the known
canonical transformation between 3D and $BF$ gravities in the absence of
fermions. Consequently, all terms in $\egbf$ should take the form
   \begin{equation}
   \dc y  x y'
   \end{equation}
where $y$ and $y'$ are spinors in $\fbff$ and in $\fgrf$ respectively, and $x$
is any product of non-spinorial fields---possibly including contractions---in
either theory. Of course, not any such combination is valid, given that every
such product must have cohomological degree $3$ and ghost number $-1$, and every
internal index must be contracted. In fact, under these constraints there will
be at most two kinds of valid products $\dc y  x y'$. The first kind will have
$x=1$ and an appropriate distribution of contractions $\ixi$, only
consisting---up to redistribution of the $\ixi$---of the following possible
pairs $(y,y')$:
   \begin{equation}
   \begin{aligned}
   & (\anti \gamma, \varepsilon), & 
   & (\anti \varphi, \psi), &
   & (\anti \gamma, \ixi \psi), \\
   & (\ixi \anti \varphi, \anti \psi), & 
   & (\ixi^2 \anti \varphi, \anti \varepsilon), &
   & (\ixi^3 \anti \gamma, \anti \varepsilon),
   \end{aligned}
   \end{equation}
or the analogous pairings exchanging the roles of the fields in $\fgrf$ by those
in $\fbff$ and vice versa. Meanwhile, the second kind of product will have
$x\neq 1$, yet it is evident that any product of this type, to be valid, should
be obtained from a product of the first kind by replacing any number of
contractions with a product $x$ of non-spinorial fields that have the same
cohomological degree and ghost number as the power of $\ixi$ that they are
replacing. In other words, 
   \begin{equation}
   \begin{pmatrix}
   \dego x \\ \gh\;x
   \end{pmatrix}
   =
   \begin{pmatrix}
   -k \\ k
   \end{pmatrix}
   \end{equation}
for some $k\in\N$. Now, every non-spinorial field---including the contraction
$\ixi$---has its pair $(\dego \bullet, \gh \bullet)$ among the following:
   \begin{equation}\label{eq:degoGhVectors}
   v_1 = 
      \begin{pmatrix}
      0 \\ 1
      \end{pmatrix},\ 
   v_2 = 
      \begin{pmatrix}
      1 \\ 0
      \end{pmatrix},\ 
   v_3 = 
      \begin{pmatrix}
      2 \\ -1
      \end{pmatrix},\ 
   v_4 = 
      \begin{pmatrix}
      3 \\ -2
      \end{pmatrix},\ 
   v_5 = 
      \begin{pmatrix}
      -1 \\ 1
      \end{pmatrix},
   \end{equation}
which respectively correspond to the degree pairs of $c$, $A$, $\anti B$, $\anti
c$ and $\ixi$, in that order. The question, then, reduces to solving the simple
equation
   \begin{equation}
   \begin{aligned}
   k^i v_i = 0 & &\text{for} &
   & \{k^i\}_{i=1}^4\subset \N,\ k^5\in\Z,
   \end{aligned}
   \end{equation}
which holds if and only if $k^i=0$ for all $i$. This is equivalent
to saying that any valid product $\dc y  x y$ is of the first kind, that
is, a Dirac product of spinors or of contractions of those with respect to
$\xi$.
\end{proof}

\begin{samepage}
\begin{remark}
\label{rem:extensionIsEasy}
Proposition \ref{prop:extensionIsEasy} facilitates the labour notably by making
$\egbf$ include at most four terms, which moreover will show a convenient property: they
will only fix the spinorial fields and, at most, modify the expression for
$\anti \xi$ as a function of the fields in $\fbff$. Finding an appropriate
extension, then, is letting
   \begin{equation}
   \egbf = \sum_i k_i \dc {y_i} y'_i
   \end{equation}
and determining the (at most) four parameters $k_i$ that will lead to a
$\qgrf$ that on shell satisfies \eqref{eq:onShellQGR} and to an action $\sgrf$
whose classical spinorial part is 
$\frac 12 \psi\, \Dg \psi$.
\end{remark}
\end{samepage}

\begin{theorem}
\label{th:BFFtoGRF} A BV extension of 3D supergravity is provided by
the tuple $\tgrf \coloneq (\fgrf, \sgrf,  \qgrf, \omgrf)$ for
   \begin{subequations}
   \label{eq:BFFtoGRF}
   \begin{align}
   \label{eq:BFFtoGRF:1}
   \fgrf &= \fgr\times T^*[-1]\Omega(M, \Pi \mathcal S(P)), \\
   \label{eq:BFFtoGRF:2}
   \sgrf &= {\phibff}^*\sbff, \\
   \label{eq:BFFtoGRF:3}
   \omgrf &= {\phibff}^*\ombff, \\
   \label{eq:BFFtoGRF:4}
   \qgrf &= \{\bullet, \sgrf\},
   \end{align}
   \end{subequations}
where the Poisson bracket is defined by $\omgrf$ and the canonical
transformation $\phibff$ is generated by
   \begin{equation}
   \gbff = \gbf + \egbf,
   \end{equation}
where $\gbf$ is the generating function \eqref{eq:genFct0} and $\egbf$ is given
as
   \begin{equation}\label{eq:genFct}
   \egbf(q,p') = -\dc{\anti \varphi}\, \psi 
   - \dc{\anti \gamma}\, (\varepsilon - \ixi \psi)
   \end{equation}
for $q\coloneq (e, \anti \Gamma,  \xi, \anti \chi, \psi, \varepsilon)$
and $p'\coloneq (\anti B, A, \anti \tau, c, \anti \varphi, \anti \gamma)$.
\end{theorem}

\begin{proof}
This generating function leads to 
   \begin{subequations}
   \begin{align}
   &\varphi = \psi, & &\anti \varphi = \anti \psi - \ixi \anti \varepsilon, \\
   &\gamma = \varepsilon - \ixi \psi, & &\anti \gamma = \anti \varepsilon,
   \end{align}
   \end{subequations}
so following the previous Remark \ref{rem:extensionIsEasy}, we only have
to attend some of the terms in $\sbff$ \eqref{eq:SBFF} to check whether it
produces an extension of classical 3D supergravity. Firstly, since the
Definition \eqref{eq:BFto3Dfull:2} of $A$ in terms of fields in $\fgrf$ remains
unchanged, the expansion of the classical spinorial field gives
   \begin{equation}
   \frac 12 \dc \varphi \, D_A \varphi = \frac 12 \dc \psi \, \Dg \psi
   - \frac 12 \dc \psi \, (\ixi \anti e - \frac 12 \ixi^2 \anti \tau) \psi,
   \end{equation}
so indeed the classical spinorial term is recovered on shell---where, remember,
antifields are set to zero. Secondly, since the only terms mo\-di\-fying $\qgr(\xi)$
are those spinorial terms in $\sbff$ that include a factor of $\anti
\tau$---because only these  depend on $\anti \xi$, specifically through $\inv
e(\anti \xi)$---we only verify the following terms in $\lbff$:
   \begin{equation}
   \frac 12 \dc \psi \, A \psi
   + \dc{\gamma}\, \anti B \varphi
   + \frac 12 \dc \gamma \, \anti \tau \gamma
   + \dc{\anti \varphi}\, A \gamma 
   + \dc{\anti \varphi}\, c \varphi 
   + \dc{\anti \gamma}\, c \gamma.
   \end{equation}
After expansion---which is rendered explicit in Definition \ref{def:SGRF}
below---one verifies that 
   \begin{equation}
   \qgrf(\xi) = 
   \frac 12 [\xi, \xi] 
   + \frac 12 \dc \eps \, \inv e (\rho) \eps + \cdots
   \end{equation}
omitting all terms that contain antifields, so indeed $\qgrf$ satisfies
\eqref{eq:onShellQGR}. Finally, \eqref{eq:genFct} holds necessarily, since
\eqref{eq:BFFtoGRF:1} merely accounts for the incorporation of fermions, while
equations \eqref{eq:BFFtoGRF:2} to \eqref{eq:BFFtoGRF:4} follow from the
definition of a canonical transformation and the fact that $\gbff$ is a
generating function. Therefore, we have constructed a theory $\tgrf$ that is
symplectomorphic to $BF$ supergravity, and moreover $\tgrf$ produces classical 3D
supergravity on shell; in other words, $\tgrf$ is a BV extension of 3D
supergravity.
\end{proof}

\begin{remark}
Due to the fact that it only involves spinors and their contraction,
the only equation in \eqref{eq:BFto3Dfull} that the extension $\egbf$ modifies
is the one corresponding to $\anti \tau$, giving
   \begin{equation} \label{eq:tauFGRF}
   \anti \tau = \inv e \bigl(
   \anti \xi 
   - \anti e \anti \Gamma 
   - \dc {\anti \gamma} \psi 
   + \ixi \anti e \anti \chi 
   \bigr).
   \end{equation}
\end{remark}

\noindent Having established this result, we find ourselves in a position to
give an explicit form for 3D supergravity.

\begin{definition} \label{def:SGRF} We will call 3-dimensional
\index{BV!supergravity} \term{BV supergravity} the theory built in Theorem
\ref{th:BFFtoGRF}. Its action is given by 
   \begin{subequations} \label{eq:SGRF}
   \begin{equation} 
   \sgrf = \sgr + \int_M \Bigl\langle \lgrf^1 \Bigr\rangle, 
   \end{equation}
for the density
   \begin{equation}
   \lgrf^1 = 
   \frac 12 \dc \psi \, \Dg \psi 
   + \dc \psi \, \anti e \varepsilon
   + \frac 12 \dc \eps \, \anti \tau \varepsilon
   + \dc {\anti \psi} \, \qgrf(\psi)
   + \dc {\anti \eps} \, \qgrf(\eps),
   \label{eq:LGRF}
   \end{equation}
   \end{subequations}
where $\anti \tau$---as given in \eqref{eq:tauFGRF}---is kept implicit for the
sake of readability. In turn, the cohomological vector field decomposes as
   \begin{subequations}
   \begin{equation}
   \qgrf = \qgr + \ext[GR]{Q},
   \end{equation}
for an extension that acts in the following manner:
   \def\interRow{0.5em}
   \begin{align} \label{eq:QGRF}
   &\ext[GR] Q(\Gamma) = \ext[GR] Q(\chi) = 0,\\ 
   &\ext[GR]Q(e) = 
   \dc \psi \, \rho \eps 
   - \ixi \dc {\anti \psi} \, \rho \kappa
   - \frac 12 \ixi^2 \dc {\anti \psi} \, \rho \psi \\
   & \pushright{-\ixi \dc {\anti \psi} \, \rho \ixi \psi
   + \frac 12 \ixi^2 \dc {\anti \eps} \, \rho (\eps - 2 \ixi \psi),}
   \raisetag{0mm} \notag \\
   &\ext[GR]Q(\xi) =
   \frac 12 \dc \varepsilon \, \esig \varepsilon
   + \frac 12 \ixi^2 \dc {\anti \psi} \, \esig \eps
   - \frac 16 \ixi^3 \dc {\anti \eps} \, \esig (2\eps - 3 \ixi \psi),%
   \raisetag{0mm} \\ 
   &\qgrf(\psi) = 
   \chi \psi
   + \Dg \kappa  
   - \ixi \anti e \kappa 
   + \frac 12 \ixi^2 \anti e \psi
   + \frac 12 \ixi^2 \anti \tau \kappa 
   - \frac 16 \ixi^3 \anti \tau \psi,  \\
   &\qgrf(\eps) = 
   \chi \eps
   + \ixi \Dg \kappa
   - \frac 12 \ixi^2 \anti e (\eps - 2 \ixi \psi)
   + \frac 16 \ixi^3 \anti \tau (2\eps - 3 \ixi \psi), 
   \end{align}
   \end{subequations}
writing $\kappa \coloneq \kap$ and $\esig \coloneq \inv e (\rho)$.
\end{definition}

\appendix 

\section{Selected results in graded geometry}

\begin{proposition}
As their non-graded counterparts, graded Lie derivatives satisfy
   \begin{equation}
   [\ld_X, \ld_Y] = \ld_{[X,Y]}\quad\forall\ X,Y\in \mathfrak X(\M).
   \end{equation}
\end{proposition}

\begin{remark}
This is a direct consequence of their definition.
\end{remark}

\begin{proposition}
\term{Cartan's identity}---or Cartan's \term{magic formula}---relates the graded
interior derivative, the exterior derivative, and the graded Lie
derivative:\index{Cartan's identity/magic formula}
   \begin{equation}\label{eq:cartanIdentity}
   [\iota_X, d] = \ld_X.
   \end{equation}
\end{proposition}

\begin{remark}
The proof is verbatim the one used in conventional differential geometry, but
keeping track of the gradings.
\end{remark}

\begin{proposition}
For any graded vector fields $X,Y\in\mathfrak X(\M)$ and form
$\omega\in\Omega(\M)$ over a graded manifold $\M$ we have
   \begin{subequations}\label{eq:LIotaCommutator}
   \begin{equation}\label{eq:lieIota}
   [\ld_X, \iota_Y]\omega = \iota_{[X,Y]}\omega
   \end{equation}
   \end{subequations}
\end{proposition}

\begin{proof}
The behaviour of the interior and Lie derivatives, as for any derivation, will
be fully determined by their action on an arbitrary function $f \in \smooth{\M}$
and on its differential $df$. Now, on the one hand it is obvious that
   \begin{align}
   [\ld_X, \iota_Y]f = \iota_{[X,Y]}f = 0.
   \end{align}
On the other hand, recalling that $\abs{\iota_Y} = \abs Y - 1$, Cartan's
identity and $d^2 = 0$ together imply that
   \begin{subequations}
   \begin{equation}
   \begin{aligned}
   \iota_Y (\ld[X] df) 
   &= \iota_Y ([\iota_X, d] df)
   = -\sign{\abs {\iota_X}} \iota_Y d(Xf) \\
   &= \sign{\abs X} Y(X f).
   \end{aligned}
   \end{equation}
This, jointly with $\abs{\ld[X]} =\abs X$, lets us conclude that
   \begin{equation}
   \begin{aligned}
   [\ld[X], \iota_Y] df 
   &= \ld[X] (\iota_Y df) - \sign{\abs {\ld[X]} \abs {\iota_Y}} \iota_Y (\ld[X] df) \\
   &= X (Y f) - \sign{\abs X (\abs Y - 1)}\sign{\abs X} Y (X f)
   = [X, Y] f \\
   &= \iota_{[X, Y]} df,
   \end{aligned}
   \end{equation}
   \end{subequations}
thus proving the proposition.
\end{proof}

\section*{Conflict of interest and data availability}
On behalf of all authors, the corresponding author states that there is no conflict of interest and that the manuscript has no associated data.

\printbibliography
\end{document}